\documentclass[12pt]{article}
\usepackage[centertags]{amsmath}
\usepackage{amsfonts}
\usepackage{color,graphicx,shortvrb}
\usepackage{amssymb}
\usepackage{amsthm}
\usepackage{psfrag}
\usepackage{empheq}
\usepackage{newlfont}
\usepackage{epstopdf}% To incorporate .eps illustrations using PDFLaTeX, etc.
\usepackage{subfigure}% Support for small, `sub' figures and tables
\newtheorem{exe}{Example}

\newtheorem{theo}{Theorem }
\newtheorem{req}{Remark }
\newtheorem{lemma}{Lemma }

\newtheorem{defi}{Definition}

\begin{document}
% ------------------------------------------------------------------------
\normalsize

\title{A Novel Finite Time Stability Analysis of Nonlinear Fractional-Order Time Delay Systems: A Fixed Point Approach}
\author{\ Abdellatif Ben Makhlouf \thanks{Mathematics Department, College of Science, Jouf University, P.O. Box: 2014, Sakaka, Saudi Arabia and Department of Mathematics, Faculty of Sciences of Sfax,  Tunisia,\  E.mail: \texttt{benmakhloufabdellatif@gmail.com}}
}
\date {}

\maketitle
\begin{abstract}
In this article, a novel Finite Time Stability (FTS) scheme for Fractional-Order Time Delayed Systems (FOTDSs) is proposed. By exploiting the fixed point approach, sufficient conditions that guarantee the robust FTS of FOTDSs have been established. Finally, two illustrative examples are presented to validate the main result.
\end{abstract}


\section{Introduction}\noindent
Fractional-Order Systems (FOS) can be defined as nonlinear systems that are modeled by Fractional Differential Equations (FDEs), carried out with non-integer derivatives. Indeed, such system dynamics \cite{1}. Indeed, such system dynamics are described by fractional derivatives. Integrals and derivatives of fractional orders are used to demonstrate objects that can be described by power-law long-range dependence or power-law nonlocality \cite{2} or fractal properties.

Note that, the Fractional Calculus (FC) has been involved in studying and analyzing the system dynamics in many  fields like physics, electrochemistry, biology, viscoelasticity, economics, plasma turbulence models, heat conduction and chaotic systems \cite{eco1,eco3,1,eco2}. In a same context, the evolution of engineering and sciences has notably refreshing the use of the FC in numerous areas of the control theory, and this includes FTS \cite{Chen5,Lazarevic22,Naifaro,Phat,Phat3}, asymptotic stability \cite{5,3,4,6}, stabilization \cite{10}, observer design \cite{11,10} and fault estimation \cite{14,12,13}.\\
The demonstration of FTS of FOTDSs in the literature has been based on different methods and concepts such us the Gronwall inequalities \cite{ben1,Chen5,Chen2,Chen1,Lazarevic22,Naifaro,Phat,Chen4,Chen3} and the Lyapunov functions \cite{Phat1,Phat2,Phat3}. As for the first group using the Gronwall inequalities, the authors in \cite{ben1}  have presented a FTS analysis of the FOTDSs using the Caputo-Katugampola derivative. Furthermore, based on the generalized Gronwall inequality, Naifar et al. in \cite{Naifaro} have described a FTS result of the FOTDSs using the Caputo Fractional Derivative (CFD). Authors in \cite{Chen2} have presented a robust FTS approach of FOTDSs. On the other hand, dealing with the second group using the Lyapunov functions, Thanh et al. in \cite{Phat3} have investigated a novel FTS analysis of FOTDSs.\\
Our work presents a novel methodology to study the FTS of FOTDSs using the fixed point approach. We will exploite a fixed point theorem in order to study finite time stability for FOTDSs. The theoretical findings are confirmed and validated by two illustrative examples.
\section{Basic results}\noindent
In this part, some theorems, lemmas and definitions are given.\\
\begin{defi}\cite{pod}
Given $0<\sigma<1$. The CFD is defined as,
\begin{equation}\label{dercap}
^CD_{r}^{\sigma} \zeta(s)=\frac{1}{\Gamma(1-\sigma)} \frac{d}{ds}\displaystyle \int_{r}^s  (s-\lambda)^{-\sigma} \big(\zeta(\lambda)-\zeta(r)\big)d\lambda.
\end{equation}
\end{defi}
\begin{defi}\cite{pod}
The Mittag-Leffler Function (MLF) is given by the following expression:
$$E_{\sigma}(t)=\displaystyle\sum_{b\geq 0} \frac{t^b}{\Gamma(b \sigma +1)},$$
where $\sigma>0$, $t\in \mathbb{C}$.\\
\end{defi}
\begin{req}
$E_{\sigma}(t)$ is an increasing function on $\mathbb{R}_+$.
\end{req}
\begin{req}
The function $\psi(s)=E_{\sigma}\big(\theta (s-r)^\sigma\big)$ satisfies
$^CD_{r}^{\sigma}\psi(s)=\theta \psi(s),$ and
$\frac{1}{\Gamma(\sigma)} \int_{r}^s \;(s-\lambda)^{\sigma-1} \psi(\lambda)  d\lambda=\frac{1}{\theta}\big( \psi(s)-1\big)$, where $\theta\in \mathbb{R}^*$.
\end{req}
\begin{req}
Note that, in order to prove the existence of the global solutions of FDEs and FOTDSs, authors in \cite{Cong,Doan2019} used MLF. 
\end{req}

\begin{theo}\cite{t}(Generalized Taylor's formulat)\label{t}
Let $0<\sigma< 1$. Assume that ${}^CD_{r_1}^{m \sigma} \vartheta \in C\big([r_1,r_2] \big)$, for each $m\in \{0,1,...,s\}$, with $s\in \mathbb{N}^*$, then we have
$$\vartheta(x)=\displaystyle \sum_{m=0}^{s-1}{}^CD_{r_1}^{m\sigma} \vartheta(r_1) \frac{(x-r_1)^{m \sigma}}{\Gamma(m \sigma+1)}+{}^CD_{r_1}^{s\sigma} \vartheta(c) \frac{(x-r_1)^{s \sigma}}{\Gamma(s \sigma+1)},$$
 with $c\in [r_1,x]$, for each $x\in (r_1,r_2]$.
\end{theo}

\begin{lemma}\label{t1}
For $0<\sigma< 1$ and $\theta>0$, we have
$$\frac{t^\sigma}{E_{\sigma}\big(\theta t^\sigma\big)}\leq \frac{\Gamma( \sigma+1)}{\theta},\ \forall t\geq0.$$
\end{lemma}
\begin{proof}
Using Theorem \ref{t} for the function $\vartheta(t)=E_{\sigma}\big(\theta t^\sigma\big)$, we get
$$\vartheta(t)=1+\frac{\theta t^\sigma}{\Gamma( \sigma+1)}+\frac{\theta^2 t^{2\sigma} \vartheta(c)}{\Gamma(2 \sigma+1)},$$
with $c\in [0,t]$, for each $t>0$.\\
Therefore
$$\vartheta(t)\geq 1+\frac{\theta t^\sigma}{\Gamma( \sigma+1)},$$
for all $t\geq 0$.\\
Then
$$1\geq \frac{1}{\vartheta(t)}+\frac{\theta t^\sigma}{\vartheta(t) \Gamma( \sigma+1)},$$
for all $t\geq 0$.\\
Hence
$$\frac{t^\sigma}{E_{\sigma}\big(\theta t^\sigma\big)}\leq \frac{\Gamma( \sigma+1)}{\theta},$$
for all $t\geq 0$.\\
\end{proof}

\begin{defi}
	A mapping $\varpi:\mathbf{Y}\times \mathbf{Y}\rightarrow [0,\infty]$ is called a generalized metric on a nonempty set $\mathbf{Y}$ if:
	\begin{description}
		\item[K1] $\varpi(\gamma_1,\gamma_2)=0$ if and only if $\gamma_1=\gamma_2$;
		\item[K2] $\varpi(\gamma_1,\gamma_2)=\varpi(\gamma_2,\gamma_1)$ for all $\gamma_1,\gamma_2\in Y$;
		\item[K3] $\varpi(\gamma_1,\gamma_3)\leq \varpi(\gamma_1,\gamma_2)+\varpi(\gamma_2,\gamma_3)$ for all $\gamma_1,\gamma_2,\gamma_3\in Y$.
	\end{description}
\end{defi}
The below theorem describes a basic result of the fixed point theory.
\begin{theo}\label{FXPthem1}\cite{Diaz1968}
Suppose that $(\mathbf{Y},\varpi)$ is a generalized complete metric space. Let $\Psi: \mathbf{Y}\rightarrow \mathbf{Y}$ is a strictly contractive operator with $C<1$. If one can find a nonnegative integer $j_0$ such that $\varpi(\Psi^{j_0+1}y_0,\Psi^{j_0}y_0)<\infty$ for some $y_0\in \mathbf{Y}$, then:
	\begin{description}
		\item[(a)] $\Psi^n y_0$ converges to a fixed point $y_1$ of $\Psi$;
		\item[(b)] $y_1$ is the unique fixed point of $\Psi$ in $\mathbf{Y}^*:=\{y_2\in \mathbf{Y}: \varpi(\Psi^{j_0} y_0,y_2)<\infty\}$;
		\item[(c)] If $y_2\in \mathbf{Y}^*$, then $\varpi(y_2,y_1)\leq \frac{1}{1-C} \varpi(\Psi y_2,y_2)$.
	\end{description}
\end{theo}
\par
A class of FOS with time delay is considered as follows:
\begin{equation}\label{e2}
^CD_{t_0}^{\beta} x(t)=A_0 x(t)+A_1 x(t-g(t))+A_2d(t)+f(t,x(t),x(t-g(t)),d(t)),\ t\geq t_0,
\end{equation}
with the initial condition $x(s)=\nu(s),$ $t_0-g\leq s \leq t_0$, with $0<\beta<1$, $g(t)$ is continuous, $0\leq g(t)\leq g$, $d(t)\in \mathbb{R}^p$ is the disturbance, $\nu\in C\big([t_0-g,t_0],\mathbb{R}^n\big)$, $A_0\in \mathbb{R}^{n\times n}$, $A_1\in \mathbb{R}^{n\times n}$, $A_2\in \mathbb{R}^{n\times p}$.\\
The function $f$ is continuous and satisfies:
\begin{equation}\label{H1}
\|f(s,u_1,u_2,u_3)-f(s,w_1,w_2,w_3)\|\leq \kappa(s) \big(\|u_1-w_1\|+\|u_2-w_2\|+\|u_3-w_3\|  \big),
\end{equation}
 and $f(s,0,0,0)=0$, for all $(s,u_1,u_2,u_3,w_1,w_2,w_3)\in \mathbb{R}_+\times\mathbb{R}^n \times \mathbb{R}^n \times\mathbb{R}^p\times\mathbb{R}^n\times\mathbb{R}^n\times\mathbb{R}^p$ where $\kappa$ is a continuous function.\\
The function $d$ is continuous and satisfies:
\begin{equation}\label{H2}
\exists \rho >0: \ \ \ d^T(t) d(t)\leq \rho^2.
\end{equation}
\begin{defi}\cite{Phat}
The FOS (\ref{e2}) is robustly FTS with respect to $\{\varepsilon_1, \varepsilon_2, \rho,T\}$, $\varepsilon_1 < \varepsilon_2$ if
$$\|\nu\|\leq \varepsilon_1$$
imply:
$$\|x(t)\|\leq \varepsilon_2,\ \forall t \in [t_0,T],$$
for all $d$ satisfying (\ref{H2}).
\end{defi}

\section{Stability analysis}
Let us denote $a_i=\displaystyle\max_{s\in [t_0,T]}\Big(\|A_i\|+\kappa(s)\Big)$ for $i=0,1,2$.
\begin{theo}\label{th1}
The FOS (\ref{e2}) is robustly FTS w.r.t. $\{\varepsilon_1, \varepsilon_2, \rho,T\}$, $\varepsilon_1 < \varepsilon_2$ if there exists $\eta>0$ such that
\begin{equation}\label{c1}
C(\varepsilon_1,\rho)=\Big(r_1 E_{\beta}\big((a_0+a_1+\eta)(T-t_0)^{\beta}\big)+1\Big)\varepsilon_1+r_2 E_{\beta}\big((a_0+a_1+\eta)(T-t_0)^{\beta}\big)\rho \leq \varepsilon_2,
\end{equation}
where $r_1=\displaystyle \frac{M(a_0+a_1+\eta) (a_0+a_1)}{\eta\Gamma(\beta+1)}$, $r_2=\displaystyle \frac{M(a_0+a_1+\eta) a_2}{\eta\Gamma(\beta+1)}$,\\
$M=\displaystyle\sup_{s\in [t_0,T]}\Big( \frac{(s-t_0)^{\beta}}{E_{\beta}\Big((a_0+a_1+\eta)(s-t_0)^{\beta}\Big)} \Big)$.
\end{theo}

\begin{proof}
Let $\nu\in C\big([t_0-g,t_0],\mathbb{R}^n\big)$ such that $\|\nu\|\leq \varepsilon_1$.\\
Define the metric $\varpi$ on $E=C\big([t_0-g,T],\mathbb{R}^n\big)$ by
$$\varpi(x_1,x_2)=\inf \Bigg\{b\in [0,\infty]: \frac{\|x_1(t)-x_2(t)\|}{h(t)} \leq b  , \forall t\in [t_0-g,T]\Bigg\},$$
where $h$ is defined by $h(s)=1$, for $s\in [t_0-g,t_0]$ and $h(s)=E_{\beta}\Big((a_0+a_1+\eta)(s-t_0)^{\beta}\Big)$ for $s\in [t_0,T]$.\\
As the author of \cite{J11} did in his theorem $3.1$, we get $(E,\varpi)$ is a generalized complete metric space.\\
Now, define the operator $\mathcal{V}:E\rightarrow E$ such that $(\mathcal{V}y)(t)=\nu(t)$, for $t\in [t_0-g,t_0]$ and

\begin{eqnarray}
	(\mathcal{V}y)(t)&=&\nu(t_0)+\frac{1}{\Gamma(\beta)}\int_{t_0}^{t} (t-s)^{\beta-1} \Big[A_0 y(s)+A_1 y(s-g(s))\notag\\
&+& A_2d(s)+f(s,y(s),y(s-g(s)),d(s))\Big] ds,
\end{eqnarray}
for $t\in [t_0,T]$.\\
Note that, for $y \in E$ we have $\mathcal{V}y \in E$.\\
It is easy to see that $\varpi(\mathcal{V} u_0,u_0)< \infty,$ and $\{u_1\in E: \varpi(u_0,u_1) < \infty\}=E\,\  \forall u_0\in E$.\\
Let $x_1,$ $x_2\in E$, we have $(\mathcal{V}x_1)(s)-(\mathcal{V}x_2)(s)=0$, for every $s\in [t_0-g,t_0]$.\\
For $t\in [t_0,T]$, we get
\begin{eqnarray}
	\Big\|(\mathcal{V} x_1)(t)-(\mathcal{V} x_2)(t)\Big\|&=&\bigg\| \int_{t_0}^t \frac{(t-l)^{\beta-1}}{\Gamma(\beta)} \Big[A_0 \big(x_1(l)-x_2(l)\big) +A_1 \big(x_1(l-g(l))-x_2(l-g(l))\big)\notag\\
&+& \Big(f\big(l,x_1(l),x_1(l-g(l)),d(l)\big)-f\big(l,x_2(l),x_2(l-g(l)),d(l)\big) \Big)     \Big]dl \bigg\|\notag\\
	&\leq& \int_{t_0}^t \; \frac{(t-l)^{\beta-1}}{\Gamma(\beta)} \Big[ \Big(\kappa(l) + \|A_0\|  \Big) \|x_1(l)-x_2(l)\|\notag\\
&+& \Big(\kappa(l) + \|A_1\|  \Big)
\|x_1(l-g(l))-x_2(l-g(l))\|\Big]dl \notag\\
	&\leq&  a_0 \int_{t_0}^t (t-l)^{\beta-1}\;  \frac{\| x_1(l)-x_2(l) \|}{\Gamma(\beta)} dl\notag\\
&+& a_1 \int_{t_0}^t (t-l)^{\beta-1}\;  \frac{\| x_1(l-g(l))-x_2(l-g(l)) \|}{\Gamma(\beta)} dl  \notag\\
&\leq&  \frac{a_0}{\Gamma(\beta)} \int_{t_0}^t (t-l)^{\beta-1}\;  \frac{\| x_1(l)-x_2(l) \|}{h(l)} h(l) dl\notag\\
&+& \frac{a_1}{\Gamma(\beta)} \int_{t_0}^t (t-l)^{\beta-1}\;  \frac{\| x_1(l-g(l))-x_2(l-g(l)) \|}{h(l-g(l))} h(l-g(l)) dl  \notag\\
&\leq&  \frac{a_0}{\Gamma(\beta)} \varpi(x_1,x_2) \int_{t_0}^t (t-l)^{\beta-1}\; h(l) dl \notag\\
&+& \frac{a_1}{\Gamma(\beta)}\varpi(x_1,x_2) \int_{t_0}^t (t-l)^{\beta-1}\; h(l-g(l)) dl.
\end{eqnarray}
Since $h$ is nondecreasing, so
\begin{eqnarray}
	\|(\mathcal{V} x_1)(t)-(\mathcal{V} x_2)(t)\|&\leq& \frac{(a_0+a_1)}{\Gamma(\beta)} \varpi(x_1,x_2) \int_{t_0}^t (t-s)^{\beta-1}\; h(s) ds \notag \\
&\leq& \frac{(a_0+a_1)}{(a_0+a_1+\eta)} h(t) \varpi(x_1,x_2), \,\, \text{for all}\,\, t\in [t_0,T].\notag
	\end{eqnarray}
Then
$$\varpi(\mathcal{V} x_1,\mathcal{V} x_2)\leq \frac{(a_0+a_1)}{(a_0+a_1+\eta)} \varpi(x_1,x_2).$$
Thus, $\mathcal{V}$ is a strictly contractive operator.\\
Now, consider the function $y_0$ defined by $y_0(s)=\nu(s)$, for $s\in [t_0-g,t_0]$ and $y_0(s)=\nu(t_0)$ for $s\in [t_0,T]$.\\
We have $(\mathcal{V}y_0)(s)-y_0(s)=0$, for every $s\in [t_0-g,t_0]$.\\
For $t\in [t_0,T]$, we get

\begin{eqnarray}
	\Big\|(\mathcal{V} y_0)(t)-y_0(t)\Big\|&=&\Big\|\frac{1}{\Gamma(\beta)}\int_{t_0}^{t} (t-\tau)^{\beta-1} [A_0 y_0(\tau)+A_1 y_0(\tau-g(\tau))\notag\\
&+& A_2d(\tau)+f\big(\tau,y_0(\tau),y_0(\tau-g(\tau)),d(\tau)\big)] d\tau \Big\| \notag\\
&\leq& \frac{1}{\Gamma(\beta)}\int_{t_0}^{t} (t-l)^{\beta-1} \big[\big(a_0+a_1 \big)\|\nu\|+a_2\rho \big]dl \notag\\
\notag\\
&\leq& \frac{\big[\big(a_0+a_1 \big)\|\nu\|+a_2\rho\big]}{\Gamma(\beta+1)}(t-t_0)^\beta. \notag\\
\end{eqnarray}
Therefore,
\begin{eqnarray}
	\frac{\Big\|(\mathcal{V} y_0)(t)-y_0(t)\Big\|}{h(t)}&\leq& M \frac{\big[\big(a_0+a_1 \big)\|\nu\|+a_2\rho\big]}{\Gamma(\beta+1)} ,\notag\\
\end{eqnarray}
then
$$\varpi(y_0,\mathcal{V} y_0)\leq  M \frac{\big[\big(a_0+a_1 \big)\|\nu\|+a_2\rho\big]}{\Gamma(\beta+1)}.$$
By using Theorem \ref{FXPthem1}, there exists a unique solution $x$ of \eqref{e2} with initial condition $\nu$ such that
\begin{eqnarray}
	\varpi(x,y_0)&\leq& \frac{(a_0+a_1+\eta)}{\eta} \frac{\big[\big(a_0+a_1 \big)\|\nu\|+a_2\rho\big]}{\Gamma(\beta+1)} M\notag\\
&\leq& r_1 \varepsilon_1 +r_2 \rho. \notag\\
\end{eqnarray}
Hence
$$\|x(t)-y_0(t)\| \leq  \big(r_1 \varepsilon_1 +r_2 \rho\big)  h(T), $$
for every $t\in [t_0,T].$\\
Then
\begin{eqnarray}
	\|x(t)\|&\leq& \|y_0(t)\|+ \|x(t)-y_0(t)\| ,\notag\\
&\leq& \Big(r_1 E_{\beta}\big((a_0+a_1+\eta)(T-t_0)^{\beta}\big)+1\Big)\varepsilon_1+r_2 E_{\beta}\big((a_0+a_1+\eta)(T-t_0)^{\beta}\big)\rho ,\notag\\
\end{eqnarray}
for every $t\in [t_0,T].$\\
Therefore, if (\ref{c1}) is satisfied then $\|x(t)\|\leq \varepsilon_2$, for all $t\in [t_0,T]$.
\end{proof}

\begin{req}
Note that, if we use Lemma \ref{t1}, we get
$$r_1\leq \frac{(a_0+a_1)}{\eta}$$
and
$$r_2\leq \frac{a_2}{\eta}.$$
Thus
$$C(\varepsilon_1,\rho)\leq \Big(  \frac{(a_0+a_1)}{\eta}  E_{\beta}\big((a_0+a_1+\eta)(T-t_0)^{\beta}\big)+1 \Big)\varepsilon_1 + \frac{a_2}{\eta}E_{\beta}\big((a_0+a_1+\eta)(T-t_0)^{\beta}\big) \rho. $$
Then, the assumption (\ref{c1}) can be relaxed by:
\begin{eqnarray}
D(\varepsilon_1,\rho)=\Big(  \frac{(a_0+a_1)}{\eta}  E_{\beta}\big((a_0+a_1+\eta)(T-t_0)^{\beta}\big)+1 \Big)\varepsilon_1+\frac{a_2}{\eta}E_{\beta}\big((a_0+a_1+\eta)(T-t_0)^{\beta}\big) \rho\leq \varepsilon_2.
\end{eqnarray}
\end{req}

\begin{req}
For the integer-order case, the main result remains the same by changing $\beta$ by $1$ and the MLF by the exponential function.
\end{req}
\section{Illustrative examples}
Two illustrative examples are considered to show the usefulness and interest of the main result.
\begin{exe}
Consider the FOS \eqref{e2}, with $\rho=10^{-1}$, $\beta=0.9$, $t_0=0$, $$g(s)=0.4\cos^2(s) \sin^2(s),$$  $$\nu(r)=\big(0, 0.09\big)^T,\  \text{for} \ r\in [-0.1,0],$$ $$f(s,x(s),x(s-g(s)),d(s))= \Big( \sin\big(0.01 x_2(s)\big), \sin\big(0.01 x_1(s-g(s))\big)  \Big)^T,$$
and
$$A_0=\left(
                          \begin{array}{cc}
                            0 & -2 \\
                            1 & 0 \\
                          \end{array}
                        \right),\  A_1=\left(
                          \begin{array}{cc}
                            0 & 3 \\
                            0 & 4 \\
                          \end{array}
                        \right),\  A_2=\left(
                          \begin{array}{cc}
                            0 & -0.8 \\
                            1 & 0 \\
                          \end{array}
                        \right).$$
We get $a_0=2.01$, $a_1=5.01$ and $a_2=1.01$.\\
For $\eta=1$, $\varepsilon_1=10^{-1}$, $\varepsilon_2=50$ and $T=0.385$, we get $D(\varepsilon_1,\rho)\simeq 49<\varepsilon_2.$\\
Then the FOS is robustly FTS w.r.t $\big(10^{-1},50,10^{-1},0.385\big)$.
\end{exe}

\begin{exe}
Consider the FOS \eqref{e2}, where $\rho=10^{-1}$, $t_0=0$, $\beta=0.6$, $$g(s)=0.4\cos^2(s) \sin^2(s),$$  $$\nu(r)=\big(0.06, 0.07\big)^T,\  \text{for} \ r\in [-0.1,0],$$ $$f(s,x(s),x(s-g(s)),d(s))= \Big( \sin\big(0.01 x_2(s-g(s))\big), \sin\big(0.01 x_1(s)\big)  \Big)^T,$$
and
$$A_0=\left(
                          \begin{array}{cc}
                            0 & -1 \\
                            2 & 0 \\
                          \end{array}
                        \right),\  A_1=\left(
                          \begin{array}{cc}
                            0.5 & 0 \\
                            0 & 1 \\
                          \end{array}
                        \right),\  A_2=\left(
                          \begin{array}{cc}
                            0 & 0.4 \\
                            -1 & 0 \\
                          \end{array}
                        \right).$$
We get $a_0=2.01$, $a_1=1.01$ and $a_2=1.01$.\\
For $\eta=1$, $\varepsilon_1=10^{-1}$, $\varepsilon_2=100$ and $T=0.49$, we get $D(\varepsilon_1,\rho)\simeq 97<\varepsilon_2.$\\
Then the FOS is robustly FTS w.r.t $\big(10^{-1},100,10^{-1},0.49\big)$.
\end{exe}

\section{Conclusion}
In this work, the robust FTS of FOTDSs with disturbances was studied. By proposing an approach based on the fixed point theory we have obtained a new sufficient condition for the robust FTS of such systems. Finally, two illustrative examples were presented to prove the validity of our result.

\end{document}